\documentclass[11pt]{llncs}
\usepackage[T2A]{fontenc}
\usepackage{amsmath}
\usepackage{amssymb}
\usepackage{epic}
\usepackage{gastex}
\usepackage{graphicx}
\usepackage[usenames]{color}% delete this line if you don't use colors

\DeclareSymbolFont{rsfscript}{OMS}{rsfs}{m}{n}
\DeclareSymbolFontAlphabet{\mathrsfs}{rsfscript}

\newcommand{\sw}{synchronizing word}
\newcommand{\sws}{synchronizing words}
\newcommand{\sa}{synchronizing automata}
\newcommand{\san}{synchronizing automaton}

\newtheorem{conj}{Conjecture}

\begin{document}

\title{On Carpi and Alessandro conjecture}

\titlerunning{On Carpi and Alessandro conjecture}

\author{Mikhail V. Berlinkov}

\authorrunning{M. V. Berlinkov}

\tocauthor{M. V.Berlinkov (Ekaterinburg, Russia)}

\institute{Department of Algebra and Discrete Mathematics, Ural State University,\\
620083 Ekaterinburg, Russia\\
\email{berlm@mail.ru}}

\maketitle

\begin{abstract}
The well known open \v{C}ern\'y conjecture states that each \san\
with $n$ states has a \sw\ of length at most $(n-1)^2$. On the other
hand, the best known upper bound is cubic of $n$. Recently, in the
paper \cite{CARPI1} of Alessandro and Carpi, the authors introduced
the new notion of strongly transitivity for automata and conjectured
that this property with a help of \emph{Extension} method allows to
get a quadratic upper bound for the length of the shortest \sws.
They also confirmed this conjecture for circular automata. We
disprove this conjecture and the long-standing \emph{Extension}
conjecture too. We also consider the widely used Extension method
and its perspectives.
\end{abstract}

%Сильная транзитивность и Гипотеза Черни
%   Синхронизируемость определение и актуальность
%   Гипотеза Черни и Нижняя оценка
%   Верхняя оценка, Важность Квадратичной оценки
%   Возможность Приведения к сильно-связным автоматам
%   Конструктивность доказательств, метод сжатия и расширения
%   Метод расширения - построение расширяющей последовательности слов
%   Результаты Dubuc,Kari,Beal,Rystsov circular, [one-cluster]
%   Необходимость оценки длин расширяющих слов
%   Carpi Новый подход к получению оценок длин расширяющих слов
%   сильная-транзитивность, связь с синхронизируемостью |u|+n-1
%   (n-2)(n+L-1)+1, L<n --> 2(n-2)(n-1)+1
%   Применимость к циклическим, эйлеровым
%   Гипотеза L<kn --> (n-2)(n+kn-1)+1
%   !Серия контрпримеров к этой гипотезе.
%   !Серия контрпримеров к гипотезе Расширения
%   !Перспективы метода расширения
%
%Метод расширения
%   Алгоритм расширения
%   Доказательство и оценка на длину, получаемого синх. слова
%   Необходимость оценки на длину расширяющих слова
%   m-расширяемость, гипотеза n-Расширения
%   Выводимость гип. Черни из гип. Расширения
%   !Серия контрпримеров
%   гип. kn-Расширения, kn-Сбалансированности --> Carpi
%
%Серия и доказательство
%
%Заключение
%   гип. kn-Локального-Расширения, kn-Локальной-Сбалансированности
%   Замечание про однокластерные автоматы (применимость расширяющего метода)
%   Позитивность полученных результатов
%   Наиболее вероятные гипотезы
%   Про грант

\section{Strongly transitivity and the \v{C}ern\'y conjecture}

%   Синхронизируемость определение и актуальность
Let $\mathrsfs{A}=\langle Q,\Sigma,\delta\rangle$ be a complete
\emph{deterministic finite automaton} (DFA), where $Q$ is the state
set, $\Sigma$ is the input alphabet, and $\delta:Q\times\Sigma\to Q$
is the transition function. The function $\delta$ extends uniquely
to a function $Q\times\Sigma^*\to Q$, where $\Sigma^*$ stands for
the free monoid over $\Sigma$; the latter function is still denoted
by $\delta$ and $\lambda$ denotes an empty word. Thus, each word in
$\Sigma^*$ acts on the set $Q$ via $\delta$. The DFA $\mathrsfs{A}$
is called \emph{synchronizing} if there exists a word $w\in\Sigma^*$
whose action resets $\mathrsfs{A}$, that is, leaves the automaton in
one particular state no matter which state in $Q$ it starts at:
$\delta(q,w)=\delta(q',w)$ for all $q,q'\in Q$. Any such word $w$ is
called a \emph{\sw} for $\mathrsfs{A}$. The minimum length of \sws\
for $\mathrsfs{A}$ is denoted by $\min_{synch}(\mathrsfs{A})$.

Synchronizing automata serve as transparent and natural models of
error-resistant systems in many applications (coding theory,
robotics, testing of reactive systems) and also reveal interesting
connections with symbolic dynamics and other parts of mathematics.
For a brief introduction to the theory of \sa\ we refer the reader
to the recent survey~\cite{Vo08}. Here we discuss one of the main
problem in this theory: the \v{C}ern\'y conjecture and related
problems.

%   Гипотеза Черни и Нижняя оценка
%   Верхняя оценка, Важность Квадратичной оценки
In the paper \cite{Ce64} at 1964 \v{C}ern\'y conjectured that each
\san\ with $n$ states has a \sw\ of length at most
\makebox{$(n-1)^2$}. He also presented the extremal series of the
$n$-state circular automata with a shortest \sw\ of length
$(n-1)^2$. Thus he proved the lower bound of the conjecture. The
conjecture is still open and the best known upper bound for the
length of the shortest \sw\ is $\frac{n^3-n}{6}$. Pin proved this
result at 1983 in \cite{Pi83} using combinatorial result of Frankl
\cite{Fr82}. Since the lower bound is quadratic and the upper bound
is cubic, it is of certain importance to prove a quadratic upper
bound.
%   Возможность Приведения к сильно-связным автоматам
%   Конструктивность доказательств, метод сжатия и расширения
%   Метод расширения - построение расширяющей последовательности слов
All existing methods for proving the upper bound of minimal length
of synchronizing words can be divided to <<compress>> and
<<extension>> methods. Methods of both types construct a finite
ordered collection of words $V=(v_1,v_2, \dots ,v_m)$, which
concatenation is synchronizing. Let us say that $m=|V|$ is the
\emph{size} of the collection $V$ and $L_V=\max_{i}{|v_i|}$ is the
\emph{length} of the collection $V$. The difference between these
types of methods is that the compress collection subsequently
compresses the set of states $Q$ to some state $p$, i.e
$$|Q| > |Q.v_1| > |Q.v_1v_2| > \ldots > |Q.v_1v_2 \dots v_m|=|\{p\}|,$$ while
the extension collection subsequently extends some state $p$ to the
set of states $Q$, i.e
$$|\{p\}| < |p.v_1^{-1}| < |p.v_1^{-1}v_2^{-1}| < \ldots < |p.v_1^{-1}v_2^{-1} \dots v_m^{-1}|=|Q|.$$
Since the size $m$ of the collections can not be more than $n-1$,
the proof of a quadratic upper bound can be reduced to the proof a
linear upper bound for the length of the collection $L_V$.

The compress method is used to prove the cubic upper bound
$\frac{n^3-n}{6}$ in the general case mentioned above. It is also
used to prove the \v{C}ern\'y conjecture for few <<small>> classes
of automata such as automata with zero, aperiodic automata
\cite{Traht2007} or interval automata. Since the \v{C}ern\'y
conjecture is proved for automata with zero, we assume automata is
strongly connected in the rest of the paper, otherwise the
considered problem can be reduced by using the construction of
automaton with zero (see \cite{Vo_CIAA07} for example).

%   Результаты Dubuc,Kari,Beal,Rystsov circular, [one-cluster]
%   Необходимость оценки длин расширяющих слов
The extension methods seem more productive to prove a quadratic
upper bounds. In 1998 Dubuc \cite{Du98} proved the \v{C}ern\'y
conjecture for circular automata, i.e. the automata with a letter,
which acts as a cyclic substitution. He used an extension method
combined with the skilful linear algebra techniques to prove this
result. In 2003 Kari \cite{Ka03} proved the \v{C}ern\'y conjecture
for Eurlian automata using extension method. The quadratic upper
bound was also confirmed for the one-cluster automata in the paper
\cite{BEAL1}. Let us note that it is the largest class of \sa\ with
proved quadratic upper bound.

%   Carpi Новый подход к получению оценок длин расширяющих слов
%   сильная-транзитивность, связь с синхронизируемостью |u|+n-1
In 2008 Arturo Carpi and Flavio D'Alessandro introduced the new
ideas for constructing the extension collection $V$ of linear
length. The ideas are based on the notion of the \emph{independent}
collection (or set) of words. The collection of words $W=(w_1,w_2,
\dots ,w_n)$ of the $n$-state automaton $\mathrsfs{A}$ is called
\emph{independent}, if for any two given state $s$ and $t$ there
exists an index $i$ such that $s.w_i=t$. The automaton
$\mathrsfs{A}$ is called \emph{strongly transitive}, if it admits
some independent collection of words $W$. It is easy to check, that
each synchronizing strongly connected automaton $\mathrsfs{A}$ is
strongly transitive. Moreover, if $u$ is synchronizing, then
$\mathrsfs{A}$ has an independent collection of length not more than
$|u|+n-1$.
%   Гипотеза L<kn --> (n-2)(n+kn-1)+1
%   Применимость к циклическим, эйлеровым
The authors also proved that this bound is tight and if the
$n$-state automaton $\mathrsfs{A}$ is strongly transitive with some
independent collection $W$, then it has a synchronizing word of
length not more than \makebox{$(n-2)(n+L_W-1)+1$}. Later they
conjectured that each \sa\ has an independent collection of linear
length. Formally, for some number $k>0$ the following
\emph{$kn$-Independent-Set} conjecture holds true.
\begin{conj}
\label{carpi_conj} Each strongly connected $n$-state synchronizing
automaton has an independent collection $W=(w_1,w_2, \ldots w_n)$ of
length less than $kn$.
\end{conj}
Since $k$ is a constant, this conjecture implies quadratic upper
bound of the minimal length of \sw\ for all synchronizing automata.
If the automaton is circular and $a$ denotes the circular letter,
then the independent collection $W$ can be chosen as $(\lambda, a,
a^2, \dots ,a^{n-1})$. Hence, the $1*n$-Independent-Set conjecture
 is true for circular automata. This implies the upper bound \makebox{$2(n-2)(n-1)+1$} for
this class of automata.

%   !Серия контрпримеров к этой гипотезе.
%   !Серия контрпримеров к гипотезе Расширения
%   !Перспективы метода расширения
Our paper is organized as follows. At first, in the
section~\ref{sec_method} we consider the \emph{Extension Algorithm}
in the universal form, introduce the $kn$-Extension and
$kn$-Balanced conjecture and prove that the last one implies
$kn$-Independent-Set conjecture. After this, in the
section~\ref{sec_series} we construct a series, which disproves
introduced conjectures, in particular, the $kn$-Independent-Set
conjecture of Carpi and Alessandro for each $k>0$. Finally, in the
section~\ref{sec_concl} we generalize the disproved conjectures to
the <<local>> form and discuss the perspectives of the extension
method.

\section{Extension Method}
\label{sec_method}
%
%Метод расширения
%   Алгоритм расширения
Let us consider precisely the extension method, implicity used in
the papers \cite{BEAL1,CARPI1,Du98,Ka03,Rystsov1}. In the rest of
the paper we assume $\mathrsfs{A}=\langle Q,\Sigma,\delta \rangle$
is an $n$-state strongly connected \san.

Suppose $C_s,C_e$ are some subsets of $Q$ and $v_s,v_e$ are some
words such that $Q.v_e=C_e \supseteq C_s,|C_s.v_s|=1$. Then the
following algorithm returns a \sw\ by constructing an extension
collection of words.

\bigskip
\noindent \textbf{Expansion Algorithm (EA)}
\smallskip
\hrule
\smallskip
\noindent \textbf{input} $\mathrsfs{A},C_s,C_e,v_s,v_e$
\smallskip

\noindent\textbf{initialization} $v\leftarrow v_s$\\[.1ex]
{}\phantom{\textbf{initialization}} $S\leftarrow C_s$

\smallskip

\noindent\textbf{while} $|S \bigcap C_e|< |C_e|$, find a word
$u(S)=u\in\Sigma^*$ of minimum length with
{}\phantom{\textbf{while}} $|S.u^{-1} \bigcap C_e|>|S \bigcap C_e|$; if none exists, \textbf{return} Failure\\[.1ex]
{}\phantom{\textbf{while}} $v\leftarrow uv$\\[.1ex]
{}\phantom{\textbf{while}} $S\leftarrow S.u^{-1}$

\smallskip
\noindent\textbf{return} $v_{e}v$
\smallskip
\hrule
\bigskip

%   Доказательство и оценка на длину, получаемого синх. слова
It is easy to show that the algorithm $EA$ works. At first, the
cycle iterates not more than $|C_e|-|C_f|$ times, because each
iteration expands the set $S \bigcap C_e$ to a one or more elements.
Let us show $EA$ does not fail, i.e. the word $u=u(S)$ exists in
each iteration. We now consider the one iteration of the cycle. Let
$|S\bigcap C_e|<|C_e|$. Since $\mathrsfs{A}$ is synchronizing, there
exists a \sw\ $u$, i.e. $Q.u = \{p\}$. Moreover, since
$\mathrsfs{A}$ is strongly connected, the word $u$ can be chosen to
satisfy $p \in S$. Thus the following calculations hold true.
$$Q \supseteq S.u^{-1} \supseteq p.u^{-1} = Q \Rightarrow S.u^{-1} = Q$$
$$|S.u^{-1} \bigcap C_e| = |Q \bigcap C_e| = |C_e| > |S \bigcap C_e|$$

Hence $u$ satisfies desired condition and after the last iteration
we have $C_e \subseteq S$ and $S=p.v^{-1}$ for some state $p \in Q$.
Since $C_e.{v_e}^{-1}=Q$, then
$$p.(v_ev)^{-1}=p.v^{-1}{v_e}^{-1}=S.{v_e}^{-1}=Q.$$
Hence the word $v_ev$ is synchronizing. Furthermore, since the cycle
of $EA$ iterates not more than $|C_e|-|C_f|$ times, the length of
$v_ev$ does not exceed
$$|v_e|+|v_s|+(|C_e|-|C_f|)\max_{S \subset C_e}{|u(S)|}$$

%   Необходимость оценки на длину расширяющих слова
Thus, in order to prove the upper bound, we need to estimate the
maximal possible length of the extension words $u(S)$ (the length of
the extension collection). Let us introduce one of the basic
definition in the paper.
%   m-расширяемость, гипотеза n-Расширения
\begin{definition}
The subset $S \subseteq Q$ of automaton $\mathrsfs{A}$ is called
\emph{$m$-Extendable} in the subset $C_e \subseteq Q$, if there
exists some (extension) word $v$ of length not more than $m$ such
that $$|(S \cap C_e).v^{-1}|> |S \cap C_e|.$$
\end{definition}

We now formulate the \emph{Extension} conjecture.
\begin{conj}
\label{ext_conj} Each proper subset $S$ of $Q$ is $n$-Extendable.
\end{conj}

%   Выводимость гип. Черни из гип. Расширения
Suppose the automaton $\mathrsfs{A}$ satisfies the Extension
conjecture. Let us set $C_e=Q$ and $v_e=\lambda$. Since
$\mathrsfs{A}$ is synchronizing, there exists a letter $v_s$ and a
subset $C_s$ such that \makebox{$|C_s.v_s|=1<|C_s|$}. Applying $EA$
for this input data, we get a \sw\ $v_ev$ as a result. Finally, we
have that $\mathrsfs{A}$ satisfies the \v{C}ern\'y conjecture
$$|v_ev| \leq |v_e|+|v_s|+(|C_e|-|C_s|)\max_{S\subset C_e}{|u(S)|} \leq 0+1+(n-2)n = (n-1)^2$$

%   !Серия контрпримеров
Thus the \v{C}ern\'y conjecture follows from the Extension
conjecture. This fact is used in the papers \cite{Du98} and
\cite{Ka03} to prove the \v{C}ern\'y conjecture for circular and
Eurlian automata respectively. The counterexample for Extension
conjecture is presented in the paper of Kari \cite{KariExample}.
However, the example is the $6$-state automaton, so the conjecture
is still open for $n>6$. In the next section we present a series of
counterexamples for $n>3$. We now generalize this conjecture to the
$kn$-Extension conjecture.

%   гип. kn-Расширения, kn-Сбалансированности --> Carpi
\begin{conj}
\label{ext_k_conj} Each proper subset $S$ of $Q$ is $kn$-Extendable.
\end{conj}

If the $kn$-Extension conjecture holds true for $\mathrsfs{A}$, then
$EA$ returns a \sw\ of length at most $(n-2)kn+1$. Since $k$ is a
constant, this bound is also quadratic. This conjecture is often
proved by using the following $kn$-\emph{Balanced} conjecture.

\begin{conj}
\label{balanced_k_conj} Each proper subset $S$ of $Q$ admits a word
collection $v_1,v_2 \ldots v_m$ such that $|v_i| < kn$ with the
following property.
$$\sum_{i=1}^{m}{[S.v_i^{-1}]}=m\frac{|S|}{|Q|}[Q],$$
where $[T]$ denotes the characteristic vector of the set $T$ in the
linear space $R^n$.
\end{conj}

One can prove that the $(k-1)n$-Balanced conjecture implies the
$kn$-Extension conjecture (for synchronizing automaton). Proofs of
this or equivalent facts can be found in the papers
\cite{Rystsov1,BEAL1,CARPI1}. Thus $kn$-Balanced conjecture also
implies the quadratic upper bound. The following lemma shows that
$kn$-Balanced conjecture implies $kn$-Independent-Set conjecture.

\begin{lemma}
\label{bal_imply_carpi} If synchronizing $n$-state automaton
$\mathrsfs{A}$ satisfies $kn$-Independent-Set conjecture then
$\mathrsfs{A}$ satisfies $kn$-Balanced conjecture.
\end{lemma}
\begin{proof}
Suppose $W=\{w_1,w_2, \ldots w_n\}$ is an independent set in the
automaton $\mathrsfs{A}$ of length less than $kn$, i.e. $|w_i| < kn$
and for any two given state $s$ and $t$ there exists an index $i$
such that $s.w_i=t$. Let us fix the arbitrary state $t$. Then for
each $s \in Q$ there exists an index $i$ such that $s.w_i=t$ or
equivalently \makebox{$\bigcup_{i=1}^{n}{t.w_i^{-1}}=Q$}. In the
linear form it can be written as
\makebox{$\sum_{i=1}^{n}{[t.w_i^{-1}]}=[Q]$}. Since the automaton is
deterministic then for each subset $S$ of $Q$, we have the desired
property of $kn$-Balanced conjecture
\begin{align*}
&\sum_{i=1}^{n}{[S.w_i^{-1}]}=\sum_{i=1}^{n}{\sum_{q \in
S}{[q.{w_i}^{-1}]}}=\sum_{q \in
S}{\sum_{i=1}^{n}{[q.{w_i}^{-1}]}}=&&\\
&=\sum_{q \in S}{[Q]}=|S|[Q]=n\frac{|S|}{|Q|}[Q].
\end{align*}
\end{proof}

\section{Slow extended series}
\label{sec_series}
%Серия и доказательство
%

The $2$-letter automaton $\mathrsfs{A}(m,k)=\langle Q,\Sigma,\delta
\rangle$ is drawn at the Figure~\ref{A2}. If for some state $q\in Q$
and some letter $d\in\Sigma$ there is no output edge from the state
$q$ labeled by $d$, then we assume the loop is drawn there. All such
edges are omitted for the sake of simplicity.

\begin{figure}[ht]
\begin{center}
\unitlength=1mm
%\begin{picture}(150,100)(50,-220)
\begin{picture}(168,85)(50,-85)
%\put(0,-85){\framebox(168,85){}}

\node(n0)(48.03,-60.15){$q_1$}

\node(n4)(32.03,-60.15){$q_0$}

\node(n5)(64.0,-60.15){$q_2$}

\node(n10)(48.03,-36.18){$s_1$}

\node(n11)(64.03,-36.15){$s_2$}

\node(n12)(120.0,-36.0){$s_k$}

\drawedge(n10,n5){$a$}

\drawedge(n11,n5){$a$}

\drawedge(n12,n5){$a$}

\drawedge(n4,n10){$b$}

\drawedge(n10,n11){$b$}

\drawedge(n4,n0){$a$}

\drawedge(n0,n5){$a$}

\node(n14)(83.94,-60.15){$q_3$}

\node(n15)(103.88,-60.15){$q_4$}

\node(n17)(159.82,-60.15){$q_{m}$}

\drawedge(n5,n14){$a$}

\drawedge(n14,n15){$a$}

\node(n19)(135.91,-60.15){$q_{m-1}$}

\drawedge(n19,n17){$a$}

\drawbpedge(n12,-203,71.19,n4,-217,39.59){$b$}

\node[linecolor=White,Nframe=n,Nfill=y,fillcolor=White](n112)(100.0,-36.0){$\ldots$}

\node[linecolor=White,Nframe=n,Nfill=y,fillcolor=White](n111)(80.0,-36.0){$\ldots$}

\drawedge[curvedepth=14.42](n17,n4){$a$}

\drawedge(n11,n111){$b$}

\drawedge(n112,n12){$b$}
\end{picture}
\end{center}
\caption{Automaton $\mathrsfs{A}(m,k)$} \label{A2}
\end{figure}

Let us denote by $C_b$ the set of all states unstable by $b$, i.e.
\makebox{$C_b=\{q_0,s_1,s_2,\ldots ,s_k\}$}. The following remark is
directly follows from the construction of the automaton and shows
when the letter $b$ can appear in the shortest expanding word.

\begin{remark}
\label{rem_b} Suppose $S$ is a subset of $Q$ unstable by $b$, i.e.
\makebox{$S.b^{-1} \neq S$}; then \makebox{$C_b \cap S \neq
\emptyset$} and $C_b \nsubseteq S$.
%, if $(C_b \setminus \{q_0\}) \subseteq S$, then $q_0 \notin S$.
\end{remark}

 We now formulate the main proposition about properties of
the collection of automata $\mathrsfs{A}(m,k)$.

\begin{proposition}
\label{main_prop}
\begin{enumerate}
    \item The series $\mathrsfs{B}_n=\mathrsfs{A}(n-2,1)$ is a counterexample of the Extension
conjecture for $n>3$;
    \item For each $c<2$ the series $\mathrsfs{B}_n$ is also a counterexample of the
    \makebox{$cn$-Extension} conjecture for $n>\frac{3}{2-c}$;
    \item For each $k \in \mathbb{N}$ the series $\mathrsfs{C}_n=\mathrsfs{A}(n-k,k)$ for $n>k^2$ is a counterexample of
the $kn$-Balanced conjecture and $kn$-Independent-Set conjecture of
Carpi and Alessandro, therefore.
\end{enumerate}
\end{proposition}
\begin{proof}

Consider the subset $S=C_b$. Let $v$ be a shortest word such that
$|S.v^{-1}|>|S|$, then it is easily proved by using
Remark~\ref{rem_b} that $v=a^m b a^m$ and the length of $v$ is equal
to $2m+1$. Indeed, since \makebox{$S=C_b$}, then by
Remark~\ref{rem_b} we have $v(1)=a$ and
$S_1=S.v(1)^{-1}=S.a^{-1}=\{q_m\}$. Further, since $C_b \cap S_1 =
\emptyset$, then $v(2)=a$. Applying these argumentations $m$ times,
we have
$$S_m=S.v(1 \ldots m)^{-1}=S.(a^{m})^{-1}=\{q_1,s_1,s_2 \ldots s_k\}.$$
Since $S_m.a^{-1}=\{q_0\} \subseteq S$, then $v(m+1)=b$ and
$S_m.b^{-1}=\{q_0,q_1\}$. If we repeat these arguments, we have that
$v=a^m b a^m$.

Thus $\mathrsfs{B}_n=\mathrsfs{A}(n-2,1)$ is the $n$-state automaton
and the shortest extension word for the subset $S$ is $v$ and its
length is $2m+1=2n-3$. Thus the first and the second items of the
proposition are proved.

We now consider the third one. It is clear that $\mathrsfs{C}_n$ is
a synchronizing $n$-state automaton. Arguing by contradiction,
suppose the $kn$-Balanced conjecture is true for $\mathrsfs{C}_n$
within the subset $S$. Then there exists a word collection $v_1,v_2
\ldots v_m$ such that $|v_i| \leq kn$ with the following property.
$$\sum_{i=1}^{m}{[S.v_i^{-1}]}=m\frac{|S|}{|Q|}[Q]$$
Since $|S|=k+1$, it is evident that there exists $j$ such that
$$|S.v_j^{-1} \bigcap \{q_0,q_1,q_2 \ldots q_m\}| \geq k+1.$$
Repeating the same argumentations as above for expanding $S$ in
$\{q_0,q_1,q_2 \ldots q_m\}$, we have $$|v_j|>|(a^m
b)^{k}a^m|=k(m+1)+m=k(n-k+1)+(n-k)=(k+1)n-k^2$$

Since $n>k^2$, then $|v_j|> (k+1)n-k^2 > k n$ and we get the
contradiction with the assumption. Hence, the $kn$-Balanced
conjecture is false for the $\mathrsfs{C}_n$ series. It completes
the proof of the proposition.

\end{proof}

\section{Conclusions}
\label{sec_concl}
%Заключение
%   гип. kn-Локального-Расширения, kn-Локальной-Сбалансированности
%   Замечание про однокластерные автоматы (применимость расширяющего метода)
%   Позитивность полученных результатов
%   Наиболее вероятные гипотезы
%   Про грант

In the previous section in the Proposition~\ref{main_prop} we
disproved the Extension and the $cn$-Extension conjecture for $c<2$.
Moreover, we disproved the $kn$-Balanced conjecture and the
conjecture of Alessandro and Carpi, therefore. However, it does not
mean that the Extension method can not be applied to prove the
\v{C}ern\'y conjecture or quadratic upper bound for the general
case. Our results show that we can not use these ways directly only
when $C_e=Q$. For instance, this method can be used to prove the
upper bound $2n^2-7n+7$ for automata with a connecting letter
(one-cluster automaton). Furthermore, we can generalize the
$kn$-Extension conjecture to the $kn$-Local-Extension conjecture as
follows.

\begin{conj}
\label{local_ext_k_conj} There are subsets $C_s,C_e$ and words
$v_s,v_e$ such that
$$|C_s.v_s|=1,C_s \subseteq C_e,|v_s|\leq k+kn(|C_s|-2)$$
and $C_e.v_e^{-1}=Q,|v_e| \leq kn(n-|C_e|)$ with the following
property. Each proper subset $S$ of $C_e$ is $kn$-Extendable in
$C_e$, i.e. $|S.v^{-1}\bigcap C_e|>|S\bigcap C_e|$ and $|v|\leq kn$
for some word $v$.
\end{conj}

If $kn$-Local-Extension is true for the automaton $\mathrsfs{A}$,
then it has a \sw\ of length at most
$$kn(n-|C_e|)+k+kn(|C_s|-2)+(|C_e|-|C_s|)kn=k(n-1)^2$$

Particulary, if $k=1$ the \v{C}ern\'y conjecture holds true for the
automaton $\mathrsfs{A}$. Note that the disproved $kn$-Balanced
conjecture also can be generalized to the $kn$-Local-Balanced
conjecture by the similar way for some subset $C_e$.

\begin{conj}
\label{local_bal_k_conj} Each proper subset $S$ of $C_e$ admits a
word collection $v_1,v_2 \ldots v_m$ such that $|v_i| < kn$ with the
following property.
$$\sum_{i=1}^{m}{[S.v_i^{-1}]}=m\frac{|S|}{|Q|}[Q],$$
where $[T]$ is a characteristic vector of the set $T$ in the linear
space $R^n$.
\end{conj}

The $kn$-Local-Balanced conjecture implies the main property of \\
\makebox{$(k+1)n$-Local-Extension}, i.e. each proper subset $S$ of
$C_e$ can be extended in $C_e$ by using the word $v$ of length at
most $(k+1)n$. One can easily prove this fact using ideas from the
papers \cite{Rystsov1,BEAL1,CARPI1} again.
\par
It is easy to see that the $1*n$-Local-Extension conjecture is true
for the automaton $\mathrsfs{A}(m,k)$ with \\
$C_s=\{q_0,q_1\},\ v_s=ba$ and $C_e=\{q_0,q_1, \ldots ,q_m\},\
v_e=a$. Moreover, by using $EA$ with this input we get the \sw\ $a
(ba^m)^{m-1} ba=v$ and $|v|=m^2+2=(n-k-1)^2+2$. By using
Remark~\ref{rem_b}, one can easily prove that this word is a
shortest \sw\ for this automaton.
\par
    The following remark is also trivially proved.
\begin{remark}
\label{rem_1cl} The $cn$-Local-Balanced conjecture is true for the
automaton $\mathrsfs{A}(m,n-m-1)$ with the same input for $c=1$ and
this value is the minimal with this property.
\end{remark}

In order to prove the quadratic bound for the one-cluster automata,
Beal and Perrin in the paper \cite{BEAL1} actually proved
$2n$-Local-Extension conjecture, using the $1*n$-Local-Balanced
conjecture as an auxiliary statement. Remark~\ref{rem_1cl} shows
that this way, \emph{directly} applied for this subclass of
automata, gives the order $O(2n^2)$ for the upper bound. Hence, the
upper bound $2n^2-7n+7$ is the best polynomial upper bound for the
one-cluster automata, one can achieve follow these techniques,
because $p(n)=2n^2-7n+7$ is the least polynomial function with a
first coefficient $2$ such that
$$p(2)=1=(2-1)^2,p(3)=4=(3-1)^2,$$
i.e. $p(n)$ coincides with a lower bound for the one-cluster
automata, which is not circular (see examples in \cite{TR_TESTAS}).

Nevertheless, the basic results of this paper are rejections of the
conjectures, the author wants to emphasize that it also can be
considered from the ``positive`` viewpoint, because it directs us to
the probably correct way for the proof of the quadratic bounds for
the length of the shortest synchronizing words for some subclass of
automata. Finally, remark the conjectures of $2n$-Extension,
$n$-Local-Extension and $n$-Local-Balanced seem to be most
interesting for proving or rejecting.

\end{document}